\documentclass[11pt,letterpaper]{article}

\usepackage{fix-cm}
\usepackage{amsmath,amssymb}
\usepackage{bm,bbm}
\usepackage{xspace}
\usepackage{latexsym}
\usepackage{color}
\usepackage{verbatim}
\usepackage{times}

\usepackage{ktmath}

\newtheorem{myclaim}[theorem]{Claim}

\renewcommand{\epsilon}{\vareps}

\newcommand{\dist}{\ensuremath{\mathit{dist}}\xspace}
\newcommand{\ucup}{\uplus} 
\newcommand{\gstretch}{\ensuremath{\mathsf{str}}\xspace}

\newcommand{\logoeps}{\log \mathopen{}\left(\ts\frac1{\vareps}\right)}

\newcommand{\procName}[1]{\ensuremath{{\tt #1}}\xspace}
\newcommand{\linOprName}[1]{\ensuremath{{\sf #1}}\xspace}

\usepackage[nohead,
            left=0.9in,right=0.9in,top=1in,
            footskip=0.5in,bottom=1in     
            ]{geometry}

\usepackage{enumitem}

\usepackage{booktabs}
\usepackage{color}
\usepackage{colortbl}
\usepackage{tabularx}

\usepackage{courier}
\usepackage[scaled=0.92]{helvet}

\usepackage[
colorlinks,
citecolor=blue,
pagebackref=true
]{hyperref}
\usepackage{url}

\definecolor{placeholderbg}{rgb}{0.85,0.85,0.85}

\usepackage{graphicx}

\usepackage[absolute]{textpos}

\usepackage{float}
\newfloat{acmcr}{b}{acmcr}

\date{}

\title{Near Linear-Work Parallel SDD Solvers, Low-Diameter
  Decomposition, and Low-Stretch Subgraphs\footnote{Contact Address: 5000 Forbes
    Ave. Computer Science Department. Pittsburgh, PA 15213. E-mail:
    \texttt{\{guyb, anupamg, i.koutis, glmiller, yangp, ktangwon\}@cs.cmu.edu}.}}

\author{
Guy E. Blelloch \qquad Anupam Gupta \qquad Ioannis Koutis$^\dag$
\qquad Gary L. Miller \\[0.2em] Richard Peng \qquad Kanat Tangwongsan\\[0.5em]
\it Carnegie Mellon University and $^\dag$University of Puerto Rico, Rio Piedras
}

\floatstyle{ruled}
\newfloat{algo}{tbp}{lop}[section]
\floatname{algo}{Algorithm}

\usepackage{microtype}  

\begin{document}


\maketitle

\begin{abstract}
  We present the design and analysis of a near linear-work parallel algorithm
  for solving symmetric diagonally dominant (SDD) linear systems.  On input of a
  SDD $n$-by-$n$ matrix $A$ with $m$ non-zero entries and a vector $b$, our
  algorithm computes a vector $\tilde{x}$ such that $\norm[A]{\tilde{x} - A^+b}
  \leq \vareps \cdot \norm[A]{A^+b}$ in $O(m\log^{O(1)}{n}\log
  {\frac1\epsilon})$ work and $O(m^{1/3+\theta}\log \frac1\epsilon)$ depth for
  any fixed $\theta > 0$.

  The algorithm relies on a parallel algorithm for generating low-stretch
  spanning trees or spanning subgraphs.  To this end, we first develop a
  parallel decomposition algorithm that in polylogarithmic depth and
  $\otilde(|E|)$ work\footnote{The $\otilde(\cdot)$ notion hides polylogarithmic
    factors.}, partitions a graph into components with polylogarithmic diameter
  such that only a small fraction of the original edges are between the
  components.  This can be used to generate low-stretch spanning trees with
  average stretch $O(n^{\alpha})$ in $O(n^{1+\alpha})$ work and $O(n^{\alpha})$
  depth.  Alternatively, it can be used to generate spanning subgraphs with
  polylogarithmic average stretch in $\otilde(|E|)$ work and polylogarithmic
  depth.  We apply this subgraph construction to derive a parallel linear system
  solver.
  By using this solver in known applications, our results imply improved
  parallel randomized algorithms for several problems, including single-source
  shortest paths, maximum flow, minimum-cost flow, and approximate maximum flow.
\end{abstract}

\section{Introduction}
\label{sec:intro}

Solving a system of linear equations $Ax = b$ is a fundamental
computing primitive that lies at the core of many numerical and
scientific computing algorithms, including the popular interior-point
algorithms.
The special case of symmetric diagonally dominant (SDD) systems has seen
substantial progress in recent years; in particular, the ground-breaking work of
Spielman and Teng showed how to solve SDD systems to accuracy $\vareps$ in time
$\otilde(m \log (\tfrac{1}{\vareps}))$, where $m$ is the number of non-zeros in the
$n$-by-$n$-matrix $A$.\footnote{The Spielman-Teng solver and all subsequent
  improvements are randomized algorithms. As a consequence, all algorithms
  relying on the solvers are also randomized. For simplicity, we omit standard
  complexity factors related to the probability of error.} This is
algorithmically significant since solving SDD systems has implications to
computing eigenvectors, solving flow problems, finding graph sparsifiers, and
problems in vision and graphics (see~\cite{Spielman:icm10,Teng10Survey} for
these and other applications).

In the sequential setting, the current best SDD solvers run in $O(m
\log n (\log\log n)^2\log (\tfrac1{\vareps}))$
time~\cite{KoutisMP:focs11}. However, with the exception of the
special case of planar SDD systems~\cite{KoutisM:soda07}, we know of
no previous parallel SDD solvers that perform
near-linear\footnote{i.e. linear up to polylog factors.} work and
achieve non-trivial parallelism. This raises a natural question:
\emph{Is it possible to solve an SDD linear system in $o(n)$ depth and
  $\otilde(m)$ work?}  This work answers this question affirmatively:

\begin{theorem}
  \label{thm:main}
  For any fixed $\theta>0$ and any $\epsilon > 0$, there is an
  algorithm \procName{SDDSolve} that on input an $n\times n$ SDD
  matrix $A$ with $m$ non-zero elements and a vector $b$, computes a
  vector $\tilde{x}$ such that $\norm[A]{\tilde{x} - A^+b} \leq
  \vareps \cdot \norm[A]{A^+b}$ in $O(m\log^{O(1)}{n}\log
  {\frac1\epsilon})$ work and $O(m^{1/3+\theta}\log \frac1\epsilon)$
  depth.
\end{theorem}

In the process of developing this algorithm, we give parallel
algorithms for constructing graph decompositions with strong-diameter
guarantees, and parallel algorithms to construct low-stretch spanning
trees and low-stretch ultra-sparse subgraphs, which may be of
independent interest.  An overview of these algorithms and their
underlying techniques is given in Section~\ref{sec:our-results-techn}.

\medskip
\noindent
\textbf{Some Applications.} Let us mention some of the implications of
Theorem~\ref{thm:main}, obtained by plugging it into known reductions.

\smallskip
\noindent
\emph{\;--- Construction of Spectral Sparsifiers.} Spielman and Srivastava
  \cite{SpielmanS:stoc08} showed that spectral sparsifiers can be
  constructed using $O(\log{n})$ Laplacian solves, and using our theorem
  we get spectral and cut sparsifiers in $\tilde{O}(m^{1/3+\theta})$
  depth and $\tilde{O}(m)$ work.

\smallskip
\noindent\emph{\;--- Flow Problems.} Daitsch and
Spielman~\cite{DaitschSpielman08} showed that various graph
optimization problems, such as max-flow, min-cost flow, and lossy flow
problems, can be reduced to $\otilde(m^{1/2})$
applications\footnote{here $\tilde{O}$ hides $\log{U}$ factors as
  well, where it's assumed that the edge weights are integers in the range
  $[1 \dots U]$} of SDD solves via interior point methods described in
\cite{Ye:book1997,Renegar:book01,Boyd&Vandenberghe:2004}.
Combining this with our main theorem implies that these algorithms can
be parallelized to run in $\otilde(m^{5/6+\theta})$ depth and
$\otilde(m^{3/2})$ work.  This gives the first parallel algorithm with
$o(n)$ depth which is work-efficient to within $\polylog(n)$ factors
relative to the sequential algorithm for all problems analyzed in
\cite{DaitschSpielman08}.  In some sense, the parallel bounds are more
interesting than the sequential times because in many cases the
results in \cite{DaitschSpielman08} are not the best known
sequentially (e.g. max-flow)---but do lead to the best know parallel
bounds for problems that have traditionally been hard to parallelize.
Finally, we note that although \cite{DaitschSpielman08} does not
explicitly analyze shortest path, their analysis naturally generalizes
the LP for it.

Our algorithm can also be applied in the inner loop of \cite{ChristianoKMST:arxiv10-fastflow},
yielding a $\tilde{O}(m^{5/6 +\theta}\text{poly}(\epsilon^{-1}))$ depth
and $\tilde{O}(m^{4/3}\text{poly}(\epsilon^{-1}))$ work
algorithm for finding $1 - \epsilon$ approximate maximum flows and
$1 + \epsilon$ approximate minimum cuts in undirected graphs.

\section{Preliminaries and Notation}
\label{sec:prelim-notation}



We use the notation $\otilde(f(n))$ to mean $O(f(n)\polylog(f(n)))$. We use $A
\ucup B$ to denote disjoint unions, and $[k]$ to denote the set $\{1, 2, \ldots,
k\}$.  Given a graph $G = (V,E)$, 
let $\dist(u, v)$ denote the \emph{edge-count distance} (or hop distance)
between $u$ and $v$, ignoring the edge lengths. When the graph has edge lengths
$w(e)$ (also denoted by $w_e$), let $d_G(u, v)$ denote the \emph{edge-length
  distance}, the shortest path (according to these edge lengths) between $u$ and
$v$. If the graph has unit edge lengths, the two definitions coincide. We drop
subscripts when the context is clear. We denote by $V(G)$ and $E(G)$,
respectively, the set of nodes and the set of edges, and use $n = |V(G)|$ and $m
= |E(G)|$.  For an edge $e = \{u, v\}$, the stretch of $e$ on $G'$ is
$\gstretch_{G'}(e) = d_{G'}(u, v)/w(e)$. The \emph{total stretch} of $G = (V, E,
w)$ with respect to $G'$ is $\gstretch_{G'}(E(G)) = \sum_{e \in E(G)}
\gstretch_{G'}(e)$.


Given $G = (V,E)$, a distance function $\delta$ (which is either
$\dist$ or $d$), and a partition of $V$ into $C_1 \ucup C_2 \ucup
\ldots \ucup C_p$, let $G[C_i]$ denote the induced subgraph on set
$C_i$. The \emph{weak diameter} of $C_i$ is $\max_{u,v \in C_i}
\delta_G(u,v)$, whereas the \emph{strong diameter} of $C_i$ is
$\max_{u,v \in C_i} \delta_{G[C_i]}(u,v)$; the former measures
distances in the original graph whereas the latter measures distances
within the induced subgraph.  The strong (or weak) diameter of the
partition is the maximum strong (or weak) diameter over all the
components $C_i$'s.

\smallskip
\noindent\textbf{Graph Laplacians.}
For a fixed, but arbitrary, numbering of the nodes and edges in a
graph $G = (V, E)$, the Laplacian $L_G$ of $G$ is the $|V|$-by-$|V|$
matrix given by
\begin{equation*}
  L_G(i, j) =
  \begin{cases}
    -w_{ij} & \text{ if } i \neq j\\
    \sum_{\{j,i\} \in E(G)} w_{ij} & \text{ if } i = j
  \end{cases},
\end{equation*}
When the context is clear, we use $G$ and $L_G$ interchangeably.
Given two graphs $G$ and $H$ and a scalar $\mu \in \R$, we say $G
\preceq \mu H$ if $\mu L_H - L_G$ is positive semidefinite, or
equivalently $x^{\tr} L_G x \leq \mu x^{\tr} L_H x$ for all vector $x
\in \R^{|V|}$.

\smallskip
\noindent\textbf{Matrix Norms, SDD Matrices.}
%
For a matrix $A$, we denote by $A^+$ the Moore-Penrose pseudoinverse
of $A$ (i.e., $A^+$ has the same null space as $A$ and acts as the
inverse of $A$ on its image).  Given a symmetric positive
semi-definite matrix $A$, the \emph{$A$-norm} of a vector $x$ is
defined as $\norm[A]{x} = \sqrt{x^{\tr} A x}$.  A matrix $A$ is
symmetric diagonally dominant (SDD) if it is symmetric and for all
$i$, $A_{i,i} \geq \sum_{j \neq i} |A_{i,j}|$. Solving an SDD system
reduces in $O(m)$ work and $O(\log^{O(1)} m)$ depth to solving a graph
Laplacian (a subclass of SDD matrices corresponding to undirected
weighted graphs)~\cite[Section~7.1]{Gremban-thesis}.


\smallskip
\noindent\textbf{Parallel Models.} We analyze algorithms in the
standard PRAM model, focusing on the work and depth parameters of the
algorithms. By \emph{work}, we mean the total operation count---and by
\emph{depth}, we mean the longest chain of dependencies (i.e.,
parallel time in PRAM).

\noindent\textbf{Parallel Ball Growing.} Let $B_G(s, r)$ denote the ball of
edge-count distance $r$ from a source $s$, i.e., $B_G(s, r) = \{ v \in
V(G) : \dist_{G}(s, v) \leq r\}$. We rely on an elementary form of
parallel breadth-first search to compute $B_G(s, r)$. The algorithm
visits the nodes level by level as they are encountered in the BFS
order. More precisely, level $0$ contains only the source node $s$,
level $1$ contains the neighbors of $s$, and each subsequent level
$i+1$ contains the neighbors of level $i$'s nodes that have not shown
up in a previous level. On standard parallel models (e.g., \pCRCW),
this can be computed in $O(r\log n)$ depth and $O(m' + n')$ work,
where $m'$ and $n$' are the total numbers of edges and nodes,
respectively, encountered in the
search~\cite{UllmanY:siamjc91,KleinS:jal97}.  Notice that we could
achieve this runtime bound with a variety of graph (matrix)
representations, e.g., using the compressed sparse-row (CSR) format.
Our applications apply ball growing on $r$ roughly $O(\log^{O(1)} n)$, resulting
in a small depth bound.  We remark that the idea of small-radius parallel ball
growing has previously been employed in the context of approximate shortest
paths (see, e.g.,~\cite{UllmanY:siamjc91,KleinS:jal97,Cohen:jacm00}). There is
an alternative approach of repeatedly squaring a matrix, which can give a better
depth bound for large $r$ \emph{at the expense} of a much larger work bound
(about $n^3$).

Finally, we state a tail bound which will be useful in our analysis. This bound
is easily derived from well-known facts about the tail of a hypergeometric
random
variable~\cite{Chvatal-hypgeometric,Hoeffding-full:jasa63,Skala:misc2009}.
\begin{lemma}[Hypergeometric Tail Bound]
  \label{lem:hypgtailbound}
  Let $H$ be a hypergeometric random variable denoting the number of
  red balls found in a sample of $n$ balls drawn from a total of $N$
  balls of which $M$ are red. Then, if $\mu = \expct{H} = nM/N$, then
  \begin{equation*}
    \prob{H \geq 2\mu} \;\;\leq\;\; e^{-\mu/4}
  \end{equation*}
\end{lemma}
\begin{proof}
  We apply the following theorem of
  Hoeffding~\cite{Chvatal-hypgeometric,Hoeffding-full:jasa63,Skala:misc2009}. For
  any $t > 0$,
  \begin{align*}
    \prob{H \geq \mu + tn}
    \leq
    \left(
      \Big(\frac{p}{p+t}\Big)^{p+t} \Big(\frac{1 - p}{1 - p - t} \Big)^{1 - p - t}
    \right)^n,
  \end{align*}
  where $p = \mu/n$. Using $t = p$, we have
  \begin{align*}
    \prob{H \geq 2\mu }
    &\leq \left( \Big(\frac{p}{2p}\Big)^{2p}
      \Big(\frac{1 - p}{1 - 2p} \Big)^{1 - 2p} \right)^n \\
    &\leq \left( e^{-p\ln 4} \Big(1 + \frac{p}{1 -2p}\Big)^{1-2p}\right)^n \\
    &\leq \left( e^{-p\ln 4} \cdot e^{p}\right)^n \\
    &\leq e^{-\frac14 pn},
  \end{align*}
  where we have used the fact that $1 + x \leq \exp(x)$.
\end{proof}


\section{Overview of Our  Techniques}
\label{sec:our-results-techn}
In the general solver framework of Spielman and
Teng~\cite{SpielmanTengSolver,KoutisMP:focs10}, near linear-time SDD solvers
rely on a suitable preconditioning chain of progressively smaller graphs.
Assuming that we have an algorithm for generating low-stretch spanning trees,
the algorithm as given in \cite{KoutisMP:focs10} parallelizes under the
following modifications: (i) perform the partial Cholesky factorization in
parallel and (ii) terminate the preconditioning chain with a graph that is of
size approximately $m^{1/3}$.  The details in Section~\ref{sec:SDD} are the
primary motivation of the main technical part of the work in this chapter, a
parallel implementation of a modified version of Alon et al.'s low-stretch
spanning tree algorithm~\cite{AKPW95}.

More specifically, as a first step, we find an algorithm to embed a
graph into a spanning tree with average stretch $2^{O(\sqrt{\log n
    \log\log n})}$ in $\otilde(m)$ work and $O(2^{O(\sqrt{\log n
    \log\log n})}\log \Delta)$ depth, where $\Delta$ is the ratio of
the largest to smallest distance in the graph.  The original AKPW
algorithm relies on a parallel graph decomposition scheme of
Awerbuch~\cite{Awer85}, which takes an unweighted graph and breaks it
into components with a specified diameter and few crossing edges.
While such schemes are known in the sequential setting, they do not
parallelize readily because removing edges belonging to one component
might increase the diameter or even disconnect subsequent components.
We present the first near linear-work parallel decomposition algorithm
that also gives strong-diameter guarantees, in
Section~\ref{sec:partition}, and the tree embedding results in
Section~\ref{sec:akpw}.

Ideally, we would have liked for our spanning trees to have a polylogarithmic
stretch, computable by a polylogarithmic depth, near linear-work algorithm.
However, for our solvers, we make the additional observation that we do not
really need a spanning \emph{tree} with small stretch; it suffices to give an
``ultra-sparse'' graph with small stretch, one that has only $O(m/\polylog(n))$
edges more than a tree.  Hence, we present a parallel algorithm in
Section~\ref{sec:low-stretch-subgraphs} which outputs an ultra-sparse graph with
$O(\polylog(n))$ average stretch, performing $\otilde(m)$ work with
$O(\polylog(n))$ depth.  Note that this removes the dependence of $\log \Delta$
in the depth, and reduces both the stretch and the depth from $2^{O(\sqrt{\log n
    \log\log n})}$ to $O(\polylog (n))$.\footnote{As an aside, this construction
  of low-stretch ultra-sparse graphs shows how to obtain the $\otilde(m)$-time
  linear system solver of Spielman and Teng~\cite{SpielmanTengSolver} without
  using their low-stretch spanning trees result~\cite{EEST05,AbrahamBN:focs08}.}
When combined with the aforementioned routines for constructing a SDD solver
presented in Section~\ref{sec:SDD}, this low-stretch spanning subgraph
construction yields a parallel solver algorithm.


\section{Parallel Low-Diameter Decomposition}
\label{sec:partition}

In this section, we present a parallel algorithm for partitioning a
graph into components with low (strong) diameter while cutting only a
few edges in each of the $k$ disjoint subsets of the input edges.  The
sequential version of this algorithm is at the heart of the
low-stretch spanning tree algorithm of Alon, Karp, Peleg, and West
(AKPW)~\cite{AKPW95}.

For context, notice that the outer layer of the AKPW algorithm (more
details in Section~\ref{sec:lowstretch}) can be viewed as bucketing
the input edges by weight, then partitioning and contracting them
repeatedly.  In this view, a number of edge classes are ``reduced''
simultaneously in an iteration.  Further, as we wish to output a
spanning subtree at the end, the components need to have low
strong-diameter (i.e., one could not take ``shortcuts'' through other
components). In the sequential case, the strong-diameter property is
met by removing components one after another, but this process does
not parallelize readily.  For the parallel case, we guarantee this by
growing balls from multiple sites, with appropriate ``jitters'' that
conceptually delay when these ball-growing processes start, and
assigning vertices to the first region that reaches them.  These
``jitters'' terms are crucial in controlling the probability that an
edge goes across regions.  But this probability also depends on the
number of regions that could reach such an edge.  To keep this number
small, we use a repeated sampling procedure motivated by Cohen's
$(\beta, W)$-cover construction~\cite{Cohen93}.

More concretely, we prove the following theorem:

\begin{theorem}[Parallel Low-Diameter Decomposition]
  \label{thm:graph-partition}
  Given an input graph $G = (V,E_1 \ucup \dots \ucup E_k$) with $k$
  edge classes and a ``radius'' parameter $\rho$, the algorithm
  $\procName{Partition}(G, \rho)$, upon termination, outputs a
  partition of $V$ into components $\mathcal{C} = ( C_1, C_2, \dots,
  C_p) $, each with center $s_i$ such that
  \begin{enumerate}[topsep=1pt, itemsep=1pt]
  \item the center $s_i \in C_i$ for all $i \in [p]$,

  \item for each $i$, every $u \in C_i$ satisfies $\dist_{G[C_i]}(s_i,
    u) \leq \rho$, and

  \item for all $j = 1, \dots, k$, the number of edges in $E_j$ that
    go between components is at most $|E_j| \cdot
    \frac{c_{1} \cdot k \log^3{n}}{\rho}$, where $c_{1}$ is
    an absolute constant.
  \end{enumerate}
  Furthermore, $\procName{Partition}$ runs in $O(m\log^2 n)$ expected work
  and $O(\rho\log^2{n})$ expected depth.
\end{theorem}
\subsection{Low-Diameter Decomposition for Simple Unweighted Graphs}
To prove this theorem, we begin by presenting an algorithm
\procName{splitGraph} that works with simple graphs with only
\emph{one} edge class and describe how to build on top of it an
algorithm that handles multiple edge classes.

The basic algorithm takes as input a simple, unweighted graph $G =
(V,E)$ and a radius (in hop count) parameter $\rho$ and outputs a
partition $V$ into components $C_1, \dots, C_p$, each with center
$s_i$, such that
\begin{enumerate}[topsep=1pt, itemsep=1pt, label=(P\arabic*)]
\item Each center belongs to its own component. That is, the center
  $s_i \in C_i$ for all $i \in [p]$;
\item Every component has radius at most $\rho$. That is, for each $i
  \in [p]$, every $u \in C_i$ satisfies $\dist_{G[C_i]}(s_i, u) \leq
  \rho$;
\item Given a technical condition (to be specified) that holds with
  probability at least $3/4$, the probability that an edge of the
  graph $G$ goes between components is at most $\frac{136}\rho\log^3
  n$.
\end{enumerate}
In addition, this algorithm runs in $O(m\log^2 n)$ expected work and
$O(\rho\log^2{n})$ expected depth.  (These properties should be
compared with the guarantees in Theorem~\ref{thm:graph-partition}.)

Consider the pseudocode of this basic algorithm in
Algorithm~\ref{algo:split-graph}.  The algorithm takes as input an
unweighted $n$-node graph $G$ and proceeds in $T = O(\log n)$
iterations, with the eventual goal of outputting a partition of the
graph $G$ into a collection of sets of nodes (each set of nodes is
known as a component).  Let $G^{(t)} = (V^{(t)}, E^{(t)})$ denote the
graph at the beginning of iteration $t$.  Since this graph is
unweighted, the distance in this algorithm is always the hop-count
distance $\dist(\cdot, \cdot)$.
For iteration $t = 1, \dots,
T$, the algorithm picks a set of starting centers $S^{(t)}$ to grow
balls from; as with Cohen's $(\beta, W)$-cover, the number of centers
is progressively larger with iterations, reminiscent of the doubling
trick (though with more careful handling of the growth rate), to
compensate for the balls' shrinking radius and to ensure that the
graph is fully covered.

Still within iteration $t$, it chooses a random ``jitter'' value
$\delta_s^{(t)} \in_R \{0, 1, \dots, R\}$ for each of the centers in
$S^{(t)}$ and grows a ball from each center $s$ out to radius $r^{(t)}
- \delta_s^{(t)}$, where $r^{(t)} = \frac{\rho}{2\log n}(T - t +
1)$. Let $X^{(t)}$ be the union of these balls (i.e., the nodes
``seen'' from these starting points).  In this process, the ``jitter''
should be thought of as a random amount by which we delay the
ball-growing process on each center, so that we could assign nodes to
the first region that reaches them while being in control of the
number of cross-component edges.  Equivalently, our algorithm forms
the components by assigning each vertex $u$ reachable from one of
these centers to the center that minimizes $\dist_{G^{(t)}}(u, s) +
\delta_s^{(t)}$ (ties broken in a consistent manner, e.g.,
lexicographically).  Note that because of these ``jitters,'' some
centers might not be assigned any vertex, not even itself.  For
centers that are assigned some nodes, we include their components in
the output, designating them as the components' centers.  Finally, we
construct $G^{(t+1)}$ by removing nodes that were ``seen'' in this
iteration (i.e., the nodes in $X^{(t)}$)---because they are already
part of one of the output components---and adjusting the edge set
accordingly.


\begin{algo}[th]
  \caption{\procName{splitGraph}$(G = (V, E),\rho)$ --- Split an input
    graph $G=(V,E)$ into components of hop-radius at most $\rho$.}
  \label{algo:split-graph}

  \small Let $G^{(1)} = (V^{(1)}, E^{(1)}) \gets G$.  \small Define $R
  = \rho/(2 \log{n})$.
  \small Create empty collection of components $\mathcal{C}$. \\
  \small Use $\dist^{(t)}$ as shorthand for $\dist_{G^{(t)}}$, and
  define $B^{(t)}(u, r) \eqdef B_{G^{(t)}}(u,r) = \{ v \in V^{(t)} \;\mid\;
  \dist^{(t)}(u, v) \leq r\}$.

  \medskip
  \small For $t = 1, 2, \dots, T = 2\log_2 n$,

  \begin{enumerate}[ref=\arabic*, label=\arabic*., topsep=0pt,itemsep=0pt]
  \item \label{step:sg-pick-sites} Randomly sample $S^{(t)} \subseteq V^{(t)} $, where
    $|S^{(t)|}| = \sigma_t = 12n^{t/T - 1}|V^{(t)}|\log n$, or use
    $S^{(t)} = V^{(t)}$ if $|V^{(t)}| < \sigma_t$.

  \item For each ``center'' $s \in S^{(t)}$, draw $\delta_s^{(t)}$
    uniformly at random from $\Z \cap [0, R]$.

  \item Let $r^{(t)} \gets (T - t + 1)R$.

  \item \label{part:centercreation} For each center $s \in S^{(t)}$,
     compute the ball  $B^{(t)}_s = B^{(t)}(s, r^{(t)} - \delta_s^{(t)})$.


  \item  Let $X^{(t)} = \cup_{s \in S^{(t)}}
    B^{(t)}_s$.

  \item \label{part:clusterassignment}
  Create components $\{C_s^{(t)} \;\mid \; s \in S^{(t)}\}$ by assigning
    each $u \in X^{(t)}$ to the component $C_s^{(t)}$ such that $s$
    minimizes $\dist_{G^{(t)}}(u, s) + \delta^{(t)}_s$ (breaking ties
    lexicographically).

  \item Add non-empty $C^{(t)}_s$ components to $\mathcal{C}$.

  \item Set $V^{(t+1)} \gets V^{(t)} \setminus X^{(t)}$, and let
    $G^{(t+1)} \gets G^{(t)}[V^{(t+1)}]$. Quit early if $V^{(t+1)}$ is
    empty.
  \end{enumerate}
  \small Return $\mathcal{C}$.
\end{algo}

\noindent \textbf{Analysis.} Throughout this analysis, we make
reference to various quantities in the algorithm and assume the
reader's basic familiarity with our algorithm.  We begin by proving
properties (P1)--(P2).  First, we state an easy-to-verify fact, which
follows immediately by our choice of radius and components' centers.

\begin{fact}
  \label{fact:dist1}
  If vertex $u$ lies in component $C^{(t)}_s$, then $\dist^{(t)}(s, u)
  \leq r^{(t)}$.  Moreover, $u \in B^{(t)}_s$.
\end{fact}

We also need the following lemma to argue about strong diameter.

\begin{lemma}
  \label{lem:dist2}
  If vertex $u \in C^{(t)}_s$, and vertex $v \in V^{(t)}$ lies on any
  $u$-$s$ shortest path in $G^{(t)}$, then $v \in C^{(t)}_s$.
\end{lemma}
\begin{proof}
  Since $u \in C^{(t)}_s$, Fact~\ref{fact:dist1} implies $u$ belongs
  to $B^{(t)}_s$. But $\dist^{(t)}(v, i) < \dist^{(t)}(u, i)$, and
  hence $v$ belongs to $B^{(t)}_s$ and $X^{(t)}$ as well. This implies
  that $v$ is assigned to \emph{some} component $C^{(t)}_j$; we claim
  $j = s$.

  For a contradiction, assume that $j \neq s$, and hence
  $\dist^{(t)}(v, j) + \delta^{(t)}_j \leq \dist^{(t)}(v,s) +
  \delta^{(t)}_s$. In this case $\dist^{(t)}(u,j) + \delta^{(t)}_j
  \leq \dist^{(t)}(u,v) + \dist^{(t)}(v, j) + \delta^{(t)}_j$ (by the
  triangle inequality). Now using the assumption, this expression is
  at most $\dist^{(t)}(u,v) + \dist^{(t)}(v,s) + \delta^{(t)}_s =
  \dist^{(t)}(u,s) + \delta^{(t)}_s$ (since $v$ lies on the shortest
  $u$-$s$ path). But then, $u$ would be also assigned to $C^{(t)}_j$,
  a contradiction.
\end{proof}

Hence, for each non-empty component $C^{(t)}_s$, its center~$s$ lies
within the component (since it lies on the shortest path from $s$ to
any $u \in C^{(t)}_s$), which proves (P1). Moreover, by
Fact~\ref{fact:dist1} and Lemma~\ref{lem:dist2}, the (strong) radius
is at most $TR$, proving (P2). It now remains to prove (P3), and the
work and depth bounds.

\begin{lemma}
  \label{lem:overlapbound}
  For any vertex $u \in V$, with probability at least $1 - n^{-6}$,
  there are at most $68\log^2 {n}$ pairs\footnote{In fact, for a given
    $s$, there is a unique $t$---if this $s$ is ever chosen as a
    ``starting point.''} $(s,t)$ such that $s \in S^{(t)}$ and $u \in
  B^{(t)}(s, r^{(t)})$,
\end{lemma}

We will prove this lemma in a series of claims.
\begin{myclaim}
  \label{claim:denseclobbering}
  For $t \in [T]$ and $v \in V^{(t)}$, if
  $|B^{(t)}(v, r^{(t + 1)})| \geq n^{1 - t/T}$,
  then $v \in X^{(t)}$ w.p. at least $1 - n^{-12}$.
\end{myclaim}
\begin{proof}
  First, note that for any $s \in S^{(t)}$, $r^{(t)}- \delta_s \geq
  r^{(t)} - R = r^{(t + 1)}$, and so if $s \in B^{(t)}(v, r^{(t +
    1)})$, then $v \in B^{(t)}_{s}$ and hence in $X^{(t)}$.  Therefore,
  \[
  \prob{v \in X^{(t)}} \geq \prob{S^{(t)} \cap B^{(t)}(v, r^{(t + 1)})
    \neq \emptyset},\] which is the probability that a random subset
  of $V^{(t)}$ of size $\sigma_t$ hits the ball $B^{(t)}(v, r^{(t +
    1)})$.  But, \[\prob{S^{(t)} \cap B^{(t)}(v, r^{(t + 1)}) \neq
    \emptyset} \geq 1 - \left(1 - \textstyle\frac{|B^{(t)}(v, r^{(t +
        1)})|}{|V^{(t)}|}\right)^{\sigma_t},\] which is at least $ 1 -
  n^{-12}$.
\end{proof}

\begin{myclaim}
\label{claim:manyneighbor}
  For $t \in [T]$ and $v \in V$, the number of $s \in S^{(t)}$ such
  that $v \in B^{(t)}(s, r^{(t)})$ is at most $34\log n$ w.p.\ at
  least $1 - n^{-8}$.
\end{myclaim}

\begin{proof}
  For $t = 1$, the size $\sigma_1 = O(\log n)$ and hence the claim
  follows trivially. For $t \geq 2$, we condition on all the choices
  made in rounds $1, 2, \ldots, t-2$. Note that if $v$ does not
  survive in $V^{(t-1)}$, then it does not belong to $V^{(t)}$ either,
  and the claim is immediate. So, consider two cases, depending on the
  size of the ball $B^{(t-1)}(v, r^{(t)})$ in iteration $t - 1$:\

  \noindent \emph{--- Case 1.} If $|B^{(t-1)}(v, r^{(t)})| \geq n^{1 -
    (t-1)/T}$, then by Claim 3.5, with probability at least $1 -
  n^{-12}$, we have $v \in X^{(t-1)}$, so $v$ would \emph{not} belong
  to $V^{(t)}$ and this means \textbf{no} $s \in S^{(t)}$ will satisfy $v
  \in B^{(t)}(s, r^{(t)})$, proving the claim for this case.

  \noindent \emph{--- Case 2.} Otherwise, $|B^{(t-1)}(v, r^{(t)})| <
  n^{1 - (t-1)/T}$.  We have \[ |B^{(t)}(v, r^{(t)})| \leq |B^{(t-1)}(v,
  r^{(t)})| < n^{1 - (t-1)/T} \] as $B^{(t)}(v, r^{(t)})$ $\subseteq
  B^{(t-1)}(v, r^{(t)})$.  Now let $X$ be the number of $s$ such that
  $v \in B^{(t)}(s, r^{(t)})$, so 
  $X = \sum_{s \in S^{(t)}}\onef{s \in B^{(t)}(v, r^{(t)})}$.  Over
  the random choice of $S^{(t)}$,
    \[
    \prob{s \in B^{(t)}(v, r^{(t)})} = \frac{|B^{(t)}(v, r^{(t)})|}{|V^{(t)}|}
    \leq \frac{1}{|V^{(t)}|} n^{1 - (t-1)/T},
    \]
    which gives \[ \expct{X} = \sigma_t \cdot \prob{s \in B^{(t)}(v,
      r^{(t)})} \leq 17\log n.\]

    To obtain a high probability bound for $X$, we will apply the tail
    bound in Lemma~\ref{lem:hypgtailbound}. Note that $X$ is simply a
    hypergeometric random variable with the following parameters
    setting: total balls $N = |V^{(t)}|$, red balls $M = |B^{(t)}(v,
    r^{(t)})|$, and the number balls drawn is $\sigma_t$.  Therefore,
    $\prob{X \geq 34\log n} \leq \exp\{-\frac14\cdot 34\log n\}$, so
    $X \leq 34\log n$ with probability at least $1 - n^{-8}$.

    Hence, regardless of what choices we made in rounds $1, 2, \ldots,
    t-2$, the conditional probability of seeing more than $34 \log n$
    different $s$'s is at most $n^{-8}$. Hence, we can remove the
    conditioning, and the claim follows.
\end{proof}

\begin{lemma}
  \label{lem:cutprob}
  If for each vertex $u \in V$, there are at most $68\log^2{n}$ pairs
  $(s, t)$ such that $s \in S^{(t)}$ and $u \in B^{(t)}(s, r^{(t)})$,
  then for an edge $uv$, the probability that $u$ belongs to a
  different component than $v$ is at most $68\log^2 n /
  R$.
\end{lemma}

\begin{proof}
  We define a center $s \in S^{(t)}$ as ``separating'' $u$ and $v$ if
  $ |B^{(t)}_s \cap \{u,v\}| = 1$.  Clearly, if $u,v$ lie in
  different components then there is some $t \in [T]$ and some center
  $s$ that separates them.  For a center $s \in S^{(t)}$, this can
  happen only if $\delta_s = R - \dist(s, u)$, since $\dist(s, v) \leq
  \dist(s, u) - 1$.  As there are $R$ possible values of $\delta_s$,
  this event occurs with probability at most $1/R$.  And since there
  are only $68 \log^2 n$ different centers $s$ that can possibly cut
  the edge, using a trivial union bound over them gives us an upper
  bound of $68 \log^2 n/R$ on the probability.
\end{proof}

To argue about (P3), notice that the premise to
Lemma~\ref{lem:cutprob} holds with probability exceeding $1 - o(1)
\geq 3/4$.  Combining this with Lemma~\ref{lem:overlapbound} proves
property (P3), where the technical condition is the premise to
Lemma~\ref{lem:cutprob}.

Finally, we consider the work and depth of the algorithm. These are
randomized bounds.  Each computation of $B^{(t)}(v, r^{(t)})$ can be
done using a BFS.  Since $r^{(t)} \leq \rho$, the depth is bounded by
$O(\rho \log {n})$ per iteration, resulting in $O(\rho \log^2 {n})$
after $T = O(\log n)$ iterations.  As for work, by
Lemma~\ref{lem:overlapbound}, each vertex is reached by at most
$O(\log^2{n})$ starting points, yielding a total work of
$O(m\log^2{n})$.

\subsection{Low-Diameter Decomposition for Multiple Edge Classes}

Extending the basic algorithm to support multiple edge classes is
straightforward.  The main idea is as follows. Suppose we are given a
unweighted graph $G = (V, E)$, and the edge set $E$ is composed of $k$
edge classes $E_1 \ucup \cdots \ucup E_k$.  So, if we run
\procName{splitGraph} on $G = (V,E)$ and $\rho$ treating the different
classes as one, then property (P3) indicates that each
edge---regardless of which class it came from---is separated (i.e., it
goes across components) with probability $p = \frac{136}{\rho}\log^3
n$.  This allows us to prove the following corollary, which follows
directly from Markov's inequality and the union bounds.

\begin{corollary}
  \label{cor:successprob}
  With probability at least $1/4$, for all $i \in [k]$, the number of
  edges in $E_i$ that are between components is at most $|E_i|
  \frac{272 k \log^3{n}}{\rho}$.
\end{corollary}

The corollary suggests a simple way to use \procName{splitGraph} to
provide guarantees required by Theorem~\ref{thm:graph-partition}: as
summarized in Algorithm~\ref{algo:graph-partition}, we run
\procName{splitGraph} on the input graph treating all edge classes as
one and repeat it if any of the edge classes had too many edges cut
(i.e., more than $|E_i|\frac{272 k \log^3{n}}{\rho}$). As the
corollary indicates, the number of trials is a geometric random
variable with with $p=1/4$, so in expectation, it will finish after
$4$ trials.  Furthermore, although it could go on forever in the worst
case, the probability does fall exponentially fast.

\begin{algo}[th]
  \caption{\procName{Partition}$(G = (V, E = E_1 \ucup \cdots \ucup
    E_k),\rho)$ --- Partition an input graph $G$ into components of radius
   at most $\rho$.}
  \label{algo:graph-partition}

  \begin{enumerate}[topsep=1pt, itemsep=1pt]
  \item  Let $\mathcal{C} = \procName{splitGraph}((V, \ucup E_i), \rho)$.
  \item If there is some $i$ such that $E_{i}$ has more than $|E_{i}|
    \frac{272 \cdot k \log^{3}{n}}{\rho}$ edges between components,
    start over. (Recall that $k$ was the number of edge classes.)
  \end{enumerate}
  \small Return $\mathcal{C}$.
\end{algo}

Finally, we note that properties (P1) and (P2) directly give
Theorem~\ref{thm:graph-partition}(1)--(2)---and the validation step in
\procName{Partition} ensures Theorem~\ref{thm:graph-partition}(3),
setting $c_1 = 272$. The work and depth bounds for
\procName{Partition} follow from the bounds derived for
\procName{splitGraph} and Corollary~\ref{cor:successprob}.  This
concludes the proof of Theorem~\ref{thm:graph-partition}.


\section{Parallel Low-Stretch Spanning Trees and Subgraphs}
\label{sec:lowstretch}

\newcommand{\akpwfactor}{\ensuremath{2^{O(\sqrt{\log{n} \cdot \log\log{n}})}}}

This section presents parallel algorithms for low-stretch spanning
trees and for low-stretch spanning subgraphs.  To obtain the
low-stretch spanning tree algorithm, we apply the construction of Alon
et al.~\cite{AKPW95} (henceforth, the AKPW construction), together
with the parallel graph partition algorithm from the previous section.
The resulting procedure, however, is not ideal for two reasons: the
depth of the algorithm depends on the ``spread'' $\Delta$---the ratio
between the heaviest edge and the lightest edge---and even for
polynomial spread, both the depth and the average stretch are
super-logarithmic (both of them have a \akpwfactor~term).
Fortunately, for our application, we observe that we do not need
spanning trees but merely low-stretch sparse graphs. In
Section~\ref{sec:low-stretch-subgraphs}, we describe modifications to
this construction to obtain a parallel algorithm which computes sparse
subgraphs that give us only polylogarithmic average stretch and that
can be computed in polylogarithmic depth and $\otilde(m)$ work. We
believe that this construction may be of independent interest.

\subsection{Low-Stretch Spanning Trees}
\label{sec:akpw}

Using the AKPW construction, along with the \procName{Partition}
procedure from Section~\ref{sec:partition}, we will prove the
following theorem:
\begin{theorem}[Low-Stretch Spanning Tree]
  \label{thm:parallelAKPW}
  There is an algorithm $\procName{AKPW}(G)$ which given as input a
  graph $G = (V, E, w)$, produces a spanning tree in $O(\log^{O(1)} n
  \cdot \akpwfactor \log \Delta )$ expected depth and $\otilde(m)$
  expected work such that the total stretch of all edges is bounded by
  $m \cdot \akpwfactor$.
\end{theorem}

\begin{algo}[h]
  \caption{\procName{AKPW}$(G = (V, E, w))$ --- a low-stretch spanning tree
    construction.}
  \label{algo:lowstretch-tree-AKPW}
  \small
  %





  %
  \begin{enumerate}[label=\roman*., topsep=0pt, itemsep=0pt, parsep=0pt,
    labelindent=2pt, leftmargin=15pt]

  \item Normalize the edges so that $\min \{ w(e) : e \in E \} = 1$.

  \item Let $y = 2^{\sqrt{6\log n \cdot \log\log n}}$,
    $\tau = \lceil 3 \log(n)/\log y \rceil$,
    $z = 4c_1y\tau\log^3 n$.
    Initialize $T = \emptyset$.

  \item Divide $E$ into $E_1, E_2, \dots$, where $E_i = \{e \in E
    \mid w(e) \in [z^{i-1}, z^{i})\}$. \\
    Let $E^{(1)} = E$ and $E^{(1)}_i = E_i$ for all $i$.

  \item
    For $j = 1, 2, \dots,$ until the graph is exhausted,
    \begin{enumerate}[label=\arabic*., topsep=2pt, itemsep=0pt]
    \item $(C_1, C_2, \dots, C_p) = \procName{Partition}((V^{(j)},
      \ucup_{i\leq j} E^{(j)}_i), z/4)$
    \item Add a BFS tree of each component to $T$.
    \item Define graph $(V^{(j+1)}, E^{(j+1)})$ by contracting all
      edges within the components and removing all self-loops (but
      maintaining parallel edges). Create $E_i^{(j+1)}$ from
      $E_i^{(j)}$ taking into account the contractions.
   \end{enumerate}
  \item Output the tree $T$.
  \end{enumerate}
\end{algo}

Presented in Algorithm~\ref{algo:lowstretch-tree-AKPW} is a
restatement of the AKPW algorithm, except that here we will use our
parallel low-diameter decomposition for the partition step.  In words,
iteration $j$ of Algorithm~\ref{algo:lowstretch-tree-AKPW} looks at a
graph $(V^{(j)}, E^{(j)})$ which is a minor of the original graph
(because components were contracted in previous iterations, and
because it only considers the edges in the first $j$ weight
classes). It uses $\procName{Partition}((V, \ucup_{j \leq k} E_j),
z/4)$ to decompose this graph into components such that the hop radius
is at most $z/4$ and each weight class has only $1/y$ fraction of its
edges crossing between components. (Parameters $y,z$ are defined in
the algorithm and are slightly different from the original settings in
the AKPW algorithm.) It then shrinks each of the components into a
single node (while adding a BFS tree on that component to $T$), and
iterates on this graph.  Adding these BFS trees maintains the
invariant that the set of original nodes which have been contracted
into a (super-)node in the current graph are connected in $T$; hence,
when the algorithm stops, we have a spanning tree of the original
graph---hopefully of low total stretch.

%

We begin the analysis of the total stretch and running time by proving
two useful facts:
\begin{fact}
  \label{fact:edgeclass-size}
  The number of edges $|E_i^{(j)}|$ is at most $|E_i|/y^{j -i}$.
\end{fact}

\begin{proof}
  If we could ensure that the number of weight classes in play at any
  time is at most $\tau$, the number of edges in each class would fall
  by at least a factor of $\frac{c_1 \tau \log^3 n}{z/4} = 1/y$ by
  Theorem~\ref{thm:graph-partition}(3) and the definition of $z$, and
  this would prove the fact. Now, for the first $\tau$ iterations, the
  number of weight classes is at most $\tau$ just because we consider
  only the first $j$ weight classes in iteration $j$. Now in iteration
  $\tau+1$, the number of surviving edges of $E_1$ would fall to
  $|E_1|/y^{\tau} \leq |E_1|/n^3 < 1$, and hence there would only be
  $\tau$ weight classes left. It is easy to see that this invariant
  can be maintained over the course of the algorithm.
\end{proof}


\begin{fact}
  \label{fact:akpw-comp-radius}
  In iteration $j$, the radius of a component according to edge
  weights (in the expanded-out graph) is at most $z^{j+1}$.
\end{fact}

\begin{proof}
  The proof is by induction on $j$. First, note that by
  Theorem~\ref{thm:graph-partition}(2), each of the clusters computed
  in any iteration $j$ has edge-count radius at most $z/4$. Now the
  base case $j=1$ follows by noting that each edge in $E_1$ has weight
  less than $z$, giving a radius of at most $z^2/4 < z^{j+1}$. Now
  assume inductively that the radius in iteration $j - 1$ is at most
  $z^{j}$. Now any path with $z/4$ edges from the center to some node
  in the contracted graph will pass through at most $z/4$ edges of
  weight at most $z^j$, and at most $z/4 + 1$ supernodes, each of
  which adds a distance of $2z^{j}$; hence, the new radius is at most
  $z^{j+1}/4 + (z/4 + 1)2z^j \leq z^{j+1}$ as long as $z \geq 8$.
\end{proof}

Applying these facts, we bound the total stretch of an edge class.

\begin{lemma}
  \label{lemma:class-cost}
  For any $i \geq 1$,
  $\gstretch_{T}(E_i) \leq 4y^2|E_i|(4c_1\tau\log^3{n})^{\tau+1}$.
\end{lemma}

\begin{proof}
  Let $e$ be an edge in $E_i$ contracted during iteration $j$. Since $e
  \in E_i$, we know $w(e) > z^{i-1}$. By
  Fact~\ref{fact:akpw-comp-radius}, the path connecting the two
  endpoints of $e$ in $F$ has distance at most $2z^{j+1}$.  Thus,
  $\gstretch_{T}(e) \leq 2z^{j+1}/z^{i-1} = 2z^{j -i + 2}$.
  Fact~\ref{fact:edgeclass-size} indicates that the number of such edges
  is at most $|E^{(j)}_i| \leq |E_i|/y^{j - i}$.  We conclude that
  \begin{align*}
  \gstretch_{T}(E_i)
  &\leq  \sum_{j=i}^{i+\tau-1} 2z^{j-i+2}|E_i|/y^{j - i}\\
  &\leq 4y^2|E_i|(4c_1\tau\log^3{n})^{\tau+1}
  \end{align*}
\end{proof}


\begin{proof}[of Theorem~\ref{thm:parallelAKPW}]
  Summing across the edge classes gives the promised bound on stretch.
  Now there are $\lceil \log_z \Delta \rceil$ weight classes $E_i$'s
  in all, and since each time the number of edges in a (non-empty)
  class drops by a factor of $y$, the algorithm has at most $O(\log
  \Delta + \tau)$ iterations. By Theorem~\ref{thm:graph-partition} and
  standard techniques, each iteration does $O(m\log^2 n)$ work and has
  $O(z\log^2 n) =O(\log^{O(1)} n \cdot \akpwfactor)$ depth in
  expectation.
\end{proof}

\subsection{Low-Stretch Spanning Subgraphs}
\label{sec:low-stretch-subgraphs}

We now show how to alter the parallel low-stretch spanning tree
construction from the preceding section to give a low-stretch spanning
\emph{subgraph} construction that has no dependence on the ``spread,''
and moreover has only polylogarithmic stretch.  This comes at the cost
of obtaining a sparse subgraph with $n-1 + O(m/\polylog n)$ edges instead
of a tree, but suffices for our solver application.  The two main
ideas behind these improvements are the following: Firstly, the number
of surviving edges in each weight class decreases by a logarithmic
factor in each iteration; hence, we could throw in all surviving edges
after they have been whittled down in a constant number of
iterations---this removes the factor of \akpwfactor from both the
average stretch and the depth.  Secondly, if $\Delta$ is large, we
will identify certain weight-classes with $O(m/\polylog n)$ edges, which
by setting them aside, will allow us to break up the chain of
dependencies and obtain $O(\polylog n)$ depth; these edges will be thrown
back into the final solution, adding $O(m/\polylog n)$ extra edges (which
we can tolerate) without increasing the average stretch.

\subsubsection{The First Improvement}

\newcommand{\kay}{\lambda}

Let us first show how to achieve polylogarithmic stretch with an
ultra-sparse subgraph.  Given parameters $\kay \in \Z_{> 0}$ and
$\beta \geq c_2\log^3 n$ (where $c_2 = 2\cdot
(4c_1(\kay+1))^{\frac12(\kay-1)}$), we obtain the new algorithm
$\procName{SparseAKPW}(G, \kay, \beta)$ by modifying
Algorithm~\ref{algo:lowstretch-tree-AKPW} as follows:
\begin{enumerate}[label=(\arabic*), topsep=2pt, itemsep=1pt]

\item use the altered parameters $y = \frac1{c_2}\beta/\log^3 {n}$ and
  $z = 4c_1y(\kay+1)\log^3 n$;

\item in each iteration $j$, call
  $\procName{Partition}$ with at most $\kay + 1$ edge classes---keep
  the $\kay$ classes $E_j^{(j)}, E_{j-1}^{(j)}, \ldots,
  E_{j-\kay+1}^{(j)}$, but then define a ``generic bucket'' $E_0^{(j)}
  := \cup_{j' \leq j-\kay} E_{j'}^{(j)}$ as the last part of the
  partition; and

\item finally, output not just the tree $T$ but the
  subgraph $\widehat{G} = T \cup (\cup_{i \geq 1}E^{(i+\kay)}_i)$.
\end{enumerate}

\begin{lemma}
\label{lem:sparse-akpw}
Given a graph $G$, parameters $\kay \in \Z_{> 0}$ and $\beta \geq
c_2\log^3 n$ (where $c_2 = 2\cdot (4c_1(\kay+1))^{\frac12(\kay-1)}$) the
algorithm $\procName{SparseAKPW}(G, \kay, \beta)$ outputs a subgraph of
$G$ with at most $n - 1 + m ( c_2 (\log^3 {n}/\beta))^\kay$ edges and total
stretch at most $m\beta^2 \log^{3\kay + 3} {n}$. Moreover, the expected
work is $\otilde(m)$ and expected depth is
$O((c_1\beta/c_2) \kay \log^2 n (\log \Delta + \log n))$.
\end{lemma}
\begin{proof}
  The proof parallels that for Theorem~\ref{thm:parallelAKPW}.
  Fact~\ref{fact:akpw-comp-radius} remains unchanged.  The claim from
  Fact~\ref{fact:edgeclass-size} now remains true only for $j \in \{i,
  \ldots, i + \kay - 1\}$; after that the edges in $E_i^{(j)}$ become
  part of $E_0^{(j)}$, and we only give a cumulative guarantee on the
  generic bucket. But this does hurt us: if $e \in E_i$ is contracted
  in iteration $j \leq i + \kay - 1$ (i.e., it lies within a component
  formed in iteration $j$), then $\gstretch_{\widehat{G}}(e) \leq
  2z^{j-i+2}$. And the edges of $E_i$ that survive till iteration $j
  \geq i + \kay$ have stretch $1$ because they are eventually all
  added to $\widehat{G}$; hence we do not have to worry that they
  belong to the class $E_0^{(j)}$ for those iterations. Thus,
  \[\gstretch_{\widehat{G}}(E_i) \leq \sum_{j=i}^{i + \kay -1} 2z^{j- i
    + 2}\cdot |E_i|/y^{j - i} \leq 4y^2 (\frac{z}{y})^{\kay -
    1}|E_i|.\]

  Summing across the edge classes gives $\gstretch_{\widehat{G}}(E)
  \leq 4y^2(\frac{z}y)^{\kay-1}m$, which simplifies to $O(m \beta^2
  \log^{3\kay +3} n)$.  Next, the number of edges in the output
  follows directly from the fact $T$ can have at most $n - 1$ edges,
  and the number of extra edges from each class is only a $1/y^{\kay}$
  fraction (i.e., $|E^{(i + \kay)}_i| \leq |E_i|/y^\kay$ from
  Fact~\ref{fact:edgeclass-size}).  Finally, the work remains the
  same; for each of the $(\log \Delta + \tau)$ distance scales the
  depth is still $O(z \log^2 n)$, but the new value of $z$ causes this
  to become $O((c_1 \beta/c_2) \kay \log^2 n)$.
\end{proof}

\subsubsection{The Second Improvement}

The depth of the \procName{SparseAKPW} algorithm still depends on $\log
\Delta$, and the reason is straightforward: the graph $G^{(j)}$ used in
iteration $j$ is built by taking $G^{(1)}$ and contracting edges in each
iteration---hence, it depends on all previous iterations. However, the
crucial observation is that if we had $\tau$ consecutive weight classes
$E_i$'s which are empty, we could break this chain of dependencies at this
point. However, there may be no empty weight classes; but having weight
classes with relatively few edges is enough, as we show next.

\begin{fact}
  \label{fact:twoparts}
  Given a graph $G = (V,E)$ and a subset of edges $F \subseteq E$, let
  $G' = G \setminus F$ be a potentially disconnected graph. If
  $\widehat{G}'$ is a subgraph of $G'$ with total stretch
  $\gstretch_{\widehat{G}'}(E(G')) \leq D$, then the total stretch of
$E$ on $\widehat{G} := \widehat{G}' \cup F$ is at most $|F| + D$.
\end{fact}

Consider a graph $G = (V,E, w)$ with edge weights $w(e) \geq 1$, and let
$E_i(G) := \{e \in E(G) \mid w(e) \in [z^{i-1}, z^{i})\}$ be the weight
classes.  Then, $G$ is called $(\gamma, \tau)$-\emph{well-spaced} if
there is a set of \emph{special} weight classes $\{E_i(G)\}_{i \in I}$
such that for each $i \in I$, (a) there are at most $\gamma$ weight
classes before the following special weight class $\min\{i' \in I \cup
\{\infty\} \mid i' > i\}$, and (b) the $\tau$ weight classes
$E_{i-1}(G), E_{i-2}(G), \dots, E_{i - \tau}(G)$ preceding $i$ are all
empty.

\begin{lemma}
  \label{lem:nukeintervals}
  Given any graph $G = (V,E)$, $\tau \in \Z_+$, and $\theta \leq 1$,
  there exists a graph $G' = (V,E')$ which is $(4\tau/\theta,
  \tau)$-well-spaced, and $|E' \setminus E| \leq \theta \cdot |E|$.
  Moreover, $G'$ can be constructed in $O(m)$ work and $O(\log n)$
  depth.
\end{lemma}

\begin{proof}
  Let $\delta = \frac{\log{\Delta}}{\log {z}}$; note that the edge
  classes for $G$ are $E_1, \dots, E_{\delta}$, some of which may be
  empty. Denote by $E_J$ the union $\cup_{i \in J} E_i$.
  We construct $G'$ as follows: Divide these edge classes into
  disjoint groups $J_1, J_2, \ldots \subseteq [\delta]$, where each
  group consists of $\lceil\tau/\theta\rceil$ consecutive
  classes. Within a group $J_i$, by an averaging argument, there must
  be a range $L_i \subseteq J_i$ of $\tau$ \emph{consecutive} edge
  classes that contains at most a $\theta$ fraction of all the edges
  in this group, i.e., $|E_{L_i}| \leq \theta \cdot |E_{J_i}|$ and
  $|L_i| \geq \tau$. We form $G'$ by removing these the edges in all
  these groups $L_i$'s from $G$, i.e., $G' = (V, E\setminus (\cup_i
  E_{L_i}))$. This removes only a $\theta$ fraction of all the edges
  of the graph.

  We claim $G'$ is $(4\tau/\theta, \tau)$-well-spaced. Indeed, if we
  remove the group $L_i$, then we designate the smallest $j \in
  [\delta]$ such that $j > \max\{j' \in L_i\}$ as a special bucket (if
  such a $j$ exists). Since we removed the edges in $E_{L_i}$, the
  second condition for being well-spaced follows. Moreover, the number
  of buckets between a special bucket and the following one is at most
  \[2\lceil\tau/\theta\rceil - (\tau-1) \leq 4\tau/\theta.\]  Finally,
  these computations can be done in $O(m)$ work and $O(\log n)$ depth
  using standard techniques~\cite{JaJa:book92,Leighton:book92}.
\end{proof}

\begin{lemma}
  \label{lem:wellspacedAKPW}
  Let $\tau = 3{\log n}/{\log y}$. Given a graph $G$ which is $(\gamma,
  \tau)$-well-spaced, \procName{SparseAKPW} can be computed on $G$ with
  $\otilde(m)$ work and $O(\frac{c_1}{c_2}\gamma \kay \beta \log^2 n)$
  depth.
\end{lemma}

\begin{proof}
  Since $G$ is $(\gamma, \tau)$-well-spaced, each special bucket $i
  \in I$ must be preceded by $\tau$ empty buckets. Hence, in iteration
  $i$ of \procName{SparseAKPW}, any surviving edges belong to buckets
  $E_{i - \tau}$ or smaller. However, these edges have been reduced by
  a factor of $y$ in each iteration and since $\tau > \log_y n^2$, all
  the edges have been contracted in previous iterations---i.e.,
  $E^{(i)}_\ell$ for $\ell < i$ is empty.

  Consider any special bucket $i$: we claim that we can construct the
  vertex set $V^{(i)}$ that \procName{SparseAKPW} sees at the
  beginning of iteration $i$, without having to run the previous
  iterations.  Indeed, we can just take the MST on the entire graph $G
  = G^{(1)}$, retain only the edges from buckets $E_{i-\tau}$ and
  lower, and contract the connected components of this forest to get
  $V^{(i)}$. And once we know this vertex set $V^{(i)}$, we can drop
  out the edges from $E_i$ and higher buckets which have been
  contracted (these are now self-loops), and execute iterations $i,
  i+1, \ldots$ of \procName{SparseAKPW} without waiting for the
  preceding iterations to finish. Moreover, given the MST, all this
  can be done in $O(m)$ work and $O(\log n)$ depth.

  Finally, for each special bucket $i$ in parallel, we start running
  \procName{SparseAKPW} at iteration $i$. Since there are at most
  $\gamma$ iterations until the next special bucket, the total depth
  is only $O(\gamma z \log^2 n) = O(\frac{c_1}{c_2}\gamma \kay \beta
  \log^2 n)$.
\end{proof}

\begin{theorem}[Low-Stretch Subgraphs]
  \label{thm:lowstretchsubgraph}
  Given a weighted graph $G$, $\kay \in \Z_{> 0}$, and $\beta \geq
  c_2\log^3 n$ (where $c_2 = 2\cdot
  (4c_1(\kay+1))^{\frac12(\kay-1)}$), there is an algorithm
  $\procName{LSSubgraph}(G, \beta, \kay)$ that finds a subgraph
  $\widehat{G}$ such that
  \begin{enumerate}[itemsep=0pt,topsep=1pt]
  \item $|E(\widehat{G})| \leq n - 1 + m \left( c_{\textrm{LS}}
      \frac{\log^3{n}}{\beta} \right)^\kay$
  \item The total stretch (of all $E(G)$ edges) in the subgraph
    $\widehat{G}$ is at most by $m \beta^2 \log^{3\kay + 3}n$,
  \end{enumerate}
  where $c_{\textrm{LS}}$ ($ = c_2 + 1$) is a constant.  Moreover, the procedure
  runs in $\otilde(m)$ work and $O(\kay \beta^{\kay+1} \log^{3 - 3\kay} {n})$
  depth. If $\kay = O(1)$ and $\beta = \polylog(n)$, the depth term simplifies
  to $O(\log^{O(1)} n)$.
\end{theorem}
\begin{proof}
  Given a graph $G$, we set $\tau = 3{\log n}/{\log y}$ and $\theta = (\log^3
  {n}/\beta)^\kay$, and apply Lemma~\ref{lem:nukeintervals} to delete at most
  $\theta m$ edges, and get a $(4\tau/\theta, \tau)$-well-spaced graph $G'$. Let
  $m' = |E'|$. On this graph, we run \procName{SparseAKPW} to obtain a graph
  $\widehat{G}'$ with $n - 1 + m' ( c_2 (\log^3 {n}/\beta))^\kay$ edges and
  total stretch at most $m'\beta^2 \log^{3\kay + 3} {n}$; moreover,
  Lemma~\ref{lem:wellspacedAKPW} shows this can be computed with $\otilde(m)$
  work and the depth is
  \[ O\left(\frac{c_1}{c_2}(4\tau/\theta) \kay \beta \log^2 n\right)
  = O(\kay \beta^{\kay+1} \log^{3 - 3\kay} n).\]

  Finally, we output the graph $\widehat{G} = \widehat{G}' \cup (E(G)
  \setminus E(G'))$; this gives the desired bounds on stretch and the
  number of edges as implied by Fact~\ref{fact:twoparts} and
  Lemma~\ref{lem:sparse-akpw}.
\end{proof}


\section{Parallel SDD Solver} \label{sec:SDD}

In this section, we derive a parallel solver for symmetric diagonally dominant
(SDD) linear systems, using the ingredients developed in the previous sections.
The solver follows closely the line of work of
\cite{SpielmanT:focs03,SpielmanTengSolver,KoutisM:soda07,KoutisMP:focs10}. Specifically,
we will derive a proof for the main theorem (Theorem~\ref{thm:main}), the
statement of which is reproduced below.

\begin{quote}
%
  \textbf{Theorem~\ref{thm:main}.}  For any fixed $\theta>0$ and any $\vareps >
  0$, there is an algorithm \procName{SDDSolve} that on input an SDD
  matrix $A$ and a vector $b$ computes a vector $\tilde{x}$ such that
  $\norm[A]{\tilde{x} - A^+b} \leq \vareps \cdot \norm[A]{A^+b}$ in
  $O(m\log^{O(1)}{n}\log {\frac1\vareps})$ work and
  $O(m^{1/3+\theta}\log \frac1\vareps)$ depth.
\end{quote}

In proving this theorem, we will focus on Laplacian linear systems.
As noted earlier, linear systems on SDD matrices are reducible to
systems on graph Laplacians in $O(\log (m+n))$ depth and $O(m + n)$
work~\cite{Gremban-thesis}.  Furthermore, because of the one-to-one
correspondence between graphs and their Laplacians, we will use the
two terms interchangeably.



The core of the near-linear time Laplacian solvers in
\cite{SpielmanT:focs03,SpielmanTengSolver,KoutisMP:focs10} is a
``preconditioning'' chain of progressively smaller graphs $\langle A_1
=A , A_2, \ldots, A_d\rangle$, along with a well-understood recursive
algorithm, known as recursive preconditioned Chebyshev
method---\procName{rPCh}, that traverses the levels of the chain and
for each visit at level $i<d$, performs $O(1)$ matrix-vector
multiplications with $A_i$ and other simple vector-vector
operations. Each time the algorithm reaches level $d$, it solves a
linear system on $A_d$ using a direct method.  Except for solving the
bottom-level systems, all these operations can be accomplished in
linear work and $O(\log (m+n))$ depth.  The recursion itself is based
on a simple scheme; for each visit at level $i$ the algorithm makes at
most $\kappa_i'$ recursive calls to level $i+1$, where $\kappa_i'\geq
2$ is a fixed system-independent integer. Therefore, assuming we have
computed a chain of preconditioners, the total required depth is (up
to a log) equal to the total number of times the algorithm reaches the
last (and smallest) level $A_d$.

\subsection{Parallel Construction of Solver Chain}
\label{subsec:solverchain}

The construction of the preconditioning chain in
\cite{KoutisMP:focs10} relies on a subroutine that on input a graph
$A_i$, constructs a slightly sparser graph $B_i$ which is spectrally
related to $A_i$.  This ``incremental sparsification'' routine is in
turn based on the computation of a low-stretch tree for $A_i$. The
parallelization of the low-stretch tree is actually the main obstacle
in parallelizing the whole solver presented in \cite{KoutisMP:focs10}.
Crucial to effectively applying our result in
Section~\ref{sec:lowstretch} is a simple observation that the
sparsification routine of \cite{KoutisMP:focs10} only requires a
low-stretch spanning subgraph rather than a tree.
Then, with the exception of some parameters in its construction, the
preconditioning chain remains essentially the same.

The following lemma is immediate from Section 6 of \cite{KoutisMP:focs10}.

\begin{lemma}
\label{thm:incrementalsparsify}
Given a graph $G$ and a subgraph $\widehat{G}$ of $G$ such that the total
stretch of all edges in $G$ with respect to $\widehat{G}$ is $m \cdot
S$, a parameter on condition number $\kappa$, and a success
probability $1 - 1/\xi$, there is an algorithm that constructs a graph $H$
such that
\begin{enumerate}[itemsep=0pt,topsep=1pt]

\item $G \preceq H \preceq \kappa \cdot G$, and

\item $|E(H)| = |E(\widehat{G})| + (c_{\textrm{IS}} \cdot S \log{n}
  \log{\xi})/{\kappa}$
\end{enumerate}
in $O(\log^2{n})$ depth and $O(m \log^2n)$ work, where
$c_{\textrm{IS}}$ is an absolute constant.
\end{lemma}

Although Lemma~\ref{thm:incrementalsparsify} was originally stated
with $\widehat{G}$ being a spanning tree, the proof in fact works
without changes for an arbitrary subgraph. For our purposes, $\xi$ has
to be at most $O(\log n)$ and that introduces an additional
$O(\log\log{n})$ term. For simplicity, in the rest of the section, we
will consider this as an extra $\log{n}$ factor.


\begin{lemma}
  \label{lem:precondition} Given a weighted graph $G$, parameters $\kay$
  and $\eta$ such that $\eta \geq \kay \geq 16$, we can construct in
  $O(\log^{2\eta \kay}n)$ depth and $\otilde(m)$ work another graph $H$
  such that
  \begin{enumerate}[itemsep=0pt,topsep=1pt]
  \item $G \preceq H \preceq \frac1{10}\cdot \log^{\eta \kay}{n} \cdot G$
  \item $|E(H)| \leq n - 1 + m \cdot
    c_{\textrm{PC}}/{\log^{\eta \kay - 2\eta - 4\kay} {(n)}}$,
  \end{enumerate}
  where $c_{\textrm{PC}}$ is an absolute constant.
\end{lemma}

\begin{proof}
  Let $\widehat{G} = \procName{LSSubgraph}(G, \kay, \log^{\eta}
  n)$. Then, Theorem~\ref{thm:lowstretchsubgraph} shows that
  $|E(\widehat{G})|$ is at most
  \begin{align*}
     n - 1 + m \left( \frac{c_{\textit{LS}} \cdot
        \log^3{n}}{\beta} \right)^\kay = n - 1 + m \left(
      \frac{c_{\textit{LS}}}{\log^{\eta - 3} {n}} \right)^\kay
\end{align*}
Furthermore, the total stretch of all edges in $G$ with respect to $\widehat{G}$
is at most \[ S = m \beta^2 \log^{\kay + 3}{n} \leq m \log^{2\eta + 3\kay +
  3}{n}.\] Applying Lemma~\ref{thm:incrementalsparsify} with $\kappa =
\frac1{10}\log^{\eta \kay}{n}$ gives $H$ such that $G \preceq H \preceq
\frac1{10}\log^{\eta \kay}{n} \cdot G$ and $|E(H)|$ is at most
\begin{eqnarray*}
  && n - 1 + m \cdot \left ( \frac{c_{\textit{LS}}^\kay}{\log^{\kay(\eta -
        3)}n} + \frac{10\cdot c_{\textit{IS}} \log^{2\eta + 3\kay + 5
      }{n}}{\log^{\eta \kay}{n}} \right)\\
	&\leq& n - 1 + m \cdot \frac{c_{\textit{PC}}}{\log^{\eta \kay - 2\kay - 3k - 5}{n}}\\
	&\leq& n - 1 + m \cdot \frac{c_{\textit{PC}}}{\log^{\eta \kay - 2\eta - 4\kay}{n}}.
      \end{eqnarray*}
%
%
%
\end{proof}

We now give a more precise definition of the preconditioning chain we
use for the parallel solver by giving the pseudocode for constructing
it.

\begin{definition}[Preconditioning Chain]
  Consider a chain of graphs \[\mathcal{C} = \langle A_1 =A, B_1, A_2,
  \ldots, A_d\rangle,\] and denote by $n_i$ and $m_i$ the number of
  nodes and edges of $A_i$ respectively.  We say that $\mathcal{C}$ is
  \emph{preconditioning chain} for $A$ if
  \begin{enumerate}[topsep=3pt, itemsep=2pt, parsep=0pt]
   \item $B_i = \procName{IncrementalSparsify}(A_i)$.
   \item $A_{i+1} = \procName{GreedyElimination}(B_i)$.
   \item $A_i \preceq B_i \preceq  1/10 \cdot \kappa_i A_i$, for some explicitly known integer $\kappa_i$. \footnote{The constant of $1/10$ in the condition number is introduced only
to simplify subsequent notation.}
   \end{enumerate}
\end{definition}

As noted above, the \procName{rPCh} algorithm relies on finding the
solution of linear systems on $A_d$, the bottom-level systems.  To
parallelize these solves, we make use of the following fact which can
be found in Sections~3.4. and 4.2 of~\cite{Golub96}.

\begin{fact}
  A factorization $LL^\tr$ of the pseudo-inverse of an $n$-by-$n$ Laplacian
  $A$, where $L$ is a lower triangular matrix, can be computed in
  $O(n)$ time and $O(n^3)$ work, and any solves thereafter can be done
  in $O(\log{n})$ time and $O(n^2)$ work.
\end{fact}

Note that although $A$ is not positive definite, its null space is the
space spanned by the all $1$s vector when the underlying graph is
connected. Therefore, we can in turn drop the first row and column to
obtain a semi-definite matrix on which LU factorization is numerically
stable.

The routine $\procName{GreedyElimination}$ is a partial Cholesky
factorization (for details see \cite{SpielmanTengSolver} or
\cite{KoutisMP:focs10}) on vertices of degree at most $2$. From a
graph-theoretic point of view, the routine
$\procName{GreedyElimination}$ can be viewed as simply recursively
removing nodes of degree one and splicing out nodes of degree two. The
sequential version of $\procName{GreedyElimination}$ returns a graph
with no degree $1$ or $2$ nodes. The parallel version that we present
below leaves some degree-$2$ nodes in the graph, but their number will
be small enough to not affect the complexity.



\begin{lemma}
  If $G$ has $n$ vertices and $n - 1 + m$ edges, then the procedure
  $\procName{GreedyElimination}(G)$ returns a graph with at most
  $2m-2$ nodes in $O(n + m)$ work and $O(\log n)$ depth \textbf{whp.}
\end{lemma}

\begin{proof}
  The sequential version of $\procName{GreedyElimination}(G)$ is equivalent to
  repeatedly removing degree $1$ vertices and splicing out $2$ vertices until no
  more exist while maintaining self-loops and multiple edges (see, e.g.,
  \cite{SpielmanT:focs03,SpielmanTengSolver} and \cite[Section
  2.3.4]{Koutis:thesis07}).  Thus, the problem is a slight generalization of
  parallel tree contraction~\cite{MillerR:book89}.  In the parallel version, we
  show that while the graph has more than $2m - 2$ nodes, we can efficiently
  find and eliminate a ``large'' independent set of degree two nodes, in
  addition to all degree one vertices.

  We alternate between two steps, which are equivalent to
  $\procName{Rake}$ and $\procName{Compress}$
  in \cite{MillerR:book89}, until the vertex count is at most $2m-2$: \\
  Mark an independent set of degree 2 vertices, then
  \begin{enumerate}[topsep=2pt,itemsep=1pt,parsep=2pt]
     \item Contract all degree $1$ vertices, and
     \item Compress and/or contract out the marked vertices.
  \end{enumerate}

  To find the independent set, we use a randomized marking algorithm
  on the degree two vertices (this is used in place of maximal
  independent set for work efficiency): Each degree two node flips a
  coin with probability $\frac{1}{3}$ of turning up heads; we mark a node if
  it is a heads and its neighbors either did not flip a coin or
  flipped a tail. 

  We show that the two steps above will remove a constant fraction of
  ``extra'' vertices.  Let $G$ is a multigraph with $n$ vertices and
  $m+n-1$ edges.  First, observe that if all vertices have degree at
  least three then $n \leq 2(m-1)$ and we would be finished.  So, let
  $T$ be any fixed spanning tree of $G$; let $a_1$ (resp. $a_2$) be
  the number of vertices in $T$ of degree one (resp. two) and $a_3$
  the number those of degree three or more.  Similarly, let $b_1$,
  $b_2$, and $b_3$ be the number vertices in $G$ of degree $1$, $2$,
  and at least $3$, respectively, where the degree is the vertex's
  degree in $G$.

  It is easy to check that in expectation, these two steps remove
  $b_1 + \frac{4}{27} b_2 \geq b_1  + \frac{1}{7}b_2$ vertices.
  In the following, we will show that $b_1 + \frac{1}{7}b_2 \geq \frac{1}{7} \Delta n$,
  where $\Delta n = n - (2m - 2) = n -2m +2$
  denotes the number of ``extra'' vertices in the graph.  Consider
  non-tree edges and how they are attached to the tree $T$.  Let
  $m_1$, $m_2$, and $m_3$ be the number of attachment of the following
  types, respectively:
  \begin{enumerate}[labelindent=2pt, leftmargin=15pt,
    label=(\arabic*), topsep=3pt, itemsep=1pt, parsep=1pt]

  \item an attachment to $x$, a degree 1 vertex in $T$, where $x$ has at
    least one other attachment.

  \item an attachment to $x$, a degree 1 vertex in $T$, where $x$ has no
    other attachment.

  \item an attachment to a degree $2$ vertex in $T$.
  \end{enumerate}

  As each edge is incident on two endpoints, we have $m_1 + m_2 + m_3
  \leq 2m$. Also, we can lower bound $b_1$ and $b_2$ in terms of
  $m_i$'s and $a_i$'s: we have $b_1 \geq a_1 -m_1/2 - m_2$ and $b_2
  \geq m_2 +a_2 - m_3$. This gives
  \begin{align*}
    b_1 + \ts\frac{1}{7}b_2 & \geq {\ts\frac{2}{7}}(a_1 -m_1/2 -m_2) +
   {\ts\frac{1}{7}}(m_2 + a_2 -m_3) \\
    &  =  \ts\frac{2}{7}a_1 + \ts\frac{1}{7}a_2 - {\ts\frac{1}{7}}(m_1 + m_2 + m_3) \\
    & \geq \ts\frac{2}{7} a_1 + \ts\frac{1}{7}a_2 - \ts\frac{2}{7}m.
  \end{align*}
  Consequently, $b_1 + \frac{1}{7}b_2 \geq \frac{1}{7}(2a_1 + a_2 - 2m) \geq
  \frac{1}{7}\cdot \Delta n$, where to show the last step, it suffices to
  show that $n+2 \leq 2a_1+a_2$ for a tree $T$ of $n$ nodes.  WLOG, we
  may assume that all nodes of $T$ have degree either one or three, in
  which case $2a_1 = n+2$.
  Finally, by Chernoff bounds, the algorithm will finish with high
  probability in $O(\log n)$ rounds.
%
%
%
\end{proof}

\subsection{Parallel Performance of Solver Chain}

Spielman and Teng~\cite[Section 5]{SpielmanTengSolver} gave a
(sequential) time bound for solving a linear SDD system given a
preconditioner chain.  The following lemma extends their Theorem~5.5
to give parallel runtime bounds (work and depth), as a function of
$\kappa_i$'s and $m_i$'s. We note that in the bounds below, the
$m_d^2$ term arises from the dense inverse used to solve the linear
system in the bottom level.

\begin{lemma}
\label{lem:chainrunningtime}
There is an algorithm that given a preconditioner chain $\mathcal{C} =
\langle A_1 = A, A_2, \dots, A_d \rangle$ for a matrix $A$, a vector
$b$, and an error tolerance $\vareps$, computes a vector $\tilde{x}$
such that
\[
\norm[A]{\tilde{x} - A^+b} \; \leq  \; \vareps \cdot \norm[A]{A^+b},
\]
with depth bounded by
\begin{align*}
\ts
&\Bigg(\sum\limits_{1 \leq i \leq d}  \prod\limits_{1 \leq j < i} \sqrt{\kappa_j}\Bigg)
\log{n} \logoeps
\; \leq  \;
O\Bigg( \Bigg(\prod\limits_{1 \leq j < d} \sqrt{\kappa_j} \Bigg) \log{n}
  \logoeps \Bigg)
\end{align*}
%
and work bounded by
\begin{equation*}
  \left (\sum_{1 \leq i \leq d - 1} m_i \cdot \prod_{j \leq i} \sqrt{\kappa_j}
    + m_d^2 \prod_{1 \leq j < d} \sqrt{\kappa_j} \right) \logoeps.
\end{equation*}
\end{lemma}

To reason about Lemma \ref{lem:chainrunningtime}, we will rely on the
following lemma about preconditioned Chebyshev iteration and the
recursive solves that happen at each level of the chain.  This lemma
is a restatement of Spielman and Teng's Lemma 5.3 (slightly modified
so that the $\sqrt{\kappa_i}$ does not involve a constant, which shows
up instead as constant in the preconditioner chain's definition).

\begin{lemma}
  \label{lem:recursivecheby}
  Given a preconditioner chain of length $d$, it is possible to
  construct linear operators $\linOprName{solve}_{A_i}$ for all $i
  \leq d$ such that
  \[
  (1 - e^{-2}) A_i^+ \preceq solve_{A_i} \preceq (1+e^{2})
  \]
  and $\linOprName{solve}_{A_i}$ is a polynomial of degree
  $\sqrt{\kappa_i}$ involving $\linOprName{solve}_{A_{i + 1}}$ and $4$
  matrices with $m_i$ non-zero entries (from
  $\procName{GreedyElimination}$).
\end{lemma}

Armed with this, we state and prove the following lemma:
\begin{lemma}
  \label{lem:solvedepthwork}
  For $\ell \geq 1$, given any vector $b$, the vector
  $\linOprName{solve}_{A_\ell}\cdot b$ can be computed in depth
  \[
  \log{n} \sum_{\ell \leq i \leq d}  \prod_{\ell \leq j < i} \sqrt{\kappa_j}
  \]
  and work
  \[
  \sum_{\ell \leq i \leq d - 1} m_i \cdot \prod_{\ell \leq j \leq i} \sqrt{\kappa_j}
  + m_d^2 \prod_{\ell \leq j < d} \sqrt{\kappa_j}
  \]

\end{lemma}

\begin{proof}
  The proof is by induction in decreasing order on $\ell$.  When $d =
  \ell$, all we are doing is a matrix multiplication with a dense
  inverse.  This takes $O(\log{n})$ depth and $O(m_d^2)$ work.

  Suppose the result is true for $\ell + 1$.  Then since
  $\linOprName{solve}_{A_\ell}$ can be expressed as a polynomial of
  degree $\sqrt{\kappa_\ell}$ involving an operator that is
  $\linOprName{solve}_{A_{\ell+1}}$ multiplied by at most $4$ matrices
  with $O(m_\ell)$ non-zero entries.  We have that the total depth is
  \begin{align*}
  & \log{n} \sqrt{\kappa_\ell} + \sqrt{\kappa_\ell} \cdot
         \left( \log{n} \sum_{\ell + 1 \leq i \leq d}  \prod_{\ell + 1 \leq j < i} \sqrt{\kappa_j} \right) \\
  &\;\; = \;\;  \log{n} \sum_{\ell \leq i \leq d}  \prod_{\ell \leq j < i} \sqrt{\kappa_j}
  \end{align*}
  and the total work is bounded by
  \begin{align*}
  & \sqrt{\kappa_\ell}m_\ell \;+\;
   \sqrt{\kappa_\ell} \cdot
        \left( \sum_{\ell + 1 \leq i \leq d - 1} m_i \cdot \prod_{\ell + 1 \leq j \leq i} \sqrt{\kappa_j}
           + m_d^2 \prod_{\ell + 1 \leq j < d} \sqrt{\kappa_j} \right)\\
  &\;=\;\;  \sum_{\ell \leq i \leq d - 1} m_i \cdot \prod_{\ell \leq j \leq i} \sqrt{\kappa_j}
        + m_d^2 \prod_{\ell \leq j < d} \sqrt{\kappa_j}.
  \end{align*}
\end{proof}

\begin{proof}[of Lemma~\ref{lem:chainrunningtime}]
  The $\vareps$-accuracy bound follows from applying preconditioned
  Chebyshev to $\linOprName{solve}_{A_1}$ similarly to Spielman and
  Teng's Theorem 5.5~\cite{SpielmanTengSolver}, and the running time
  bounds follow from Lemma~\ref{lem:solvedepthwork} when $\ell = 1$.
\end{proof}

\subsection{Optimizing the Chain for Depth}
\label{sec:solver}

Lemma~\ref{lem:chainrunningtime} shows that the algorithm's
performance is determined by the settings of $\kappa_i$'s and $m_i$'s;
however, as we will be using Lemma~\ref{lem:precondition}, the number
of edges $m_i$ is essentially dictated by our choice of $\kappa_i$.
We now show that if we terminate chain earlier, i.e. adjusting the dimension $A_d$ to roughly $O(m^{1/3} \log \vareps^{-1})$, we can obtain good parallel performance. As a first attempt, we will set $\kappa_i$'s uniformly:

\begin{lemma}
\label{lemma:uniformkappa}
For any fixed $\theta>0$, if we construct a preconditioner chain using
Lemma~\ref{lem:precondition} setting $\kay$ to some proper constant
greater than 21, $\eta = \kay$ and extending the sequence until $m_d
\leq m^{1/3 - \delta}$ for some $\delta$ depending on $\kay$, we get a
solver algorithm that runs in $O(m^{1/3+\theta}\log(1/\vareps))$
depth and $\otilde(m\log{1/\vareps})$ work as $\kay \rightarrow
\infty$, where $\vareps$ is the accuracy precision of the solution,
as defined in the statement of Theorem~\ref{thm:main}.
\end{lemma}

\begin{proof}
  By Lemma~\ref{thm:incrementalsparsify}, we have that $m_{i+1}$---the
  number of edges in level $i+1$---is bounded by
  \[
  O(m_i \cdot \frac{c_{\textit{PC}}}{\log^{\eta\kay - 2\eta - 4\kay}})
  = O(m_i \cdot \frac{c_{\textit{PC}}}{\log^{\kay(\kay - 6)}}),
  \]
  which can be repeatedly apply to give
  \[
  m_i \leq m \cdot \left( \frac{c_{\textit{PC}}}{\log^{\kay(\kay - 6)}{n}} \right) ^{i - 1}
  \]
Therefore, when $\kay > 12$, we have that for each $i < d$,
\begin{align*}
 m_i  \cdot \prod_{j \leq i} \sqrt{\kappa(n_j)}
&\leq  m  \cdot \left( \frac{c_{\textit{PC}}}{\log^{\kay(\kay - 6)}{n}} \right) ^{i - 1}
  \cdot \left( \sqrt{\log^{\kay^2}{n}} \right)^i \\
&=  \tilde{O}(m) \cdot \left( \frac{c_{\textit{PC}}}
  {\log^{\kay(\kay-12)/2}{n} } \right)^{i}\\
& \leq  \tilde{O}(m)
\end{align*}

Now consider the term involving $m_d$.  We have that $d$ is bounded by
\[\left( \frac{2}{3}+\delta \right) \log{m} / \log{ (\frac1{c_{\textit{PC}}}
  \log{n}^{\kay(\kay - 6)})}.\]  Combining with the $\kappa_i =
\log^{\kay^2}{n}$, we get
\begin{align*}
&\!\!\!\!\!\!\!\prod_{1 \leq j \leq d} \sqrt{\kappa(n_j)}\\
&=  \left( \log{n}^{\kay^2/2} \right)
  ^{(\frac{2}{3}+\delta)\log{m} / \log{ (c \log{n}^{\kay(\kay - 6)})}}\\
&=  \exp \left( \log\log{n} \frac{\kay^2}{2} (\frac{2}{3} + \delta)
\frac{\log{m}} {\kay(\kay-6)\log\log{n} - \log{c_{\textit{PC}}}} \right)\\
&\leq \exp \left( \log\log{n} \frac{\kay^2}{2} (\frac{2}{3} + \delta)
\frac{\log{m}} {\kay(\kay-7)\log\log{n}} \right)\\
& \qquad ( \text{since }\log{c_{PC}} \geq - \log{n} )\\
&= \exp \left(\log{n} \frac{\kay}{\kay-7} (\frac{1}{3} + \frac{\delta}{2}) \right)\\
&= O(m^{(\frac{1}{3} + \frac{\delta}{2})\frac{\kay}{\kay-7}})
\end{align*}
Since $m_d = O(m^{\frac{1}{3} - \delta})$, the total work is bounded
by
\[
O(m^{(\frac{1}{3} + \frac{\delta}{2})\frac{\kay}{\kay-7} +
  \frac{2}{3} - 2\delta}) = O(m^{1 + \frac{7}{\kay - 7} - \delta
  \frac{\kay - 14}{\kay-7}})
\]
So, setting $\delta \geq \frac{7}{\kay
  - 14}$ suffices to bound the total work by $\otilde(m)$.  And, when
$\delta$ is set to $\frac{7}{\kay - 14}$, the total parallel running
time is bounded by the number of times the last layer is called
\begin{align*}
\prod_j \sqrt{\kappa(n_j)} &\leq
 O(m^{(\frac{1}{3} + \frac{1}{2(\kay - 14)})\frac{\kay}{\kay-7}})\\
&\leq  O(m^{\frac{1}{3} + \frac{7}{\kay - 14} + \frac{\kay}{2(\kay - 14)(\kay - 7)}})\\
&\leq  O(m^{\frac{1}{3} + \frac{7}{\kay - 14} + \frac{7}{\kay - 14}})\\
&\leq  O(m^{\frac{1}{3} + \frac{14}{\kay - 14}})
~~~~~\text{when }\kay\geq 21
\end{align*}
Setting $\kay$ arbitrarily large suffices to give $O(m^{1/3 +
  \theta})$ depth.
\end{proof}

To match the promised bounds in Theorem~\ref{thm:main}, we improve the
performance by reducing the exponent on the $\log{n}$ term in the
total work from $\kay^2$ to some large fixed constant while letting
total depth still approach $O(m^{1/3+\theta})$.

\medskip


\begin{proof}[of Theorem~\ref{thm:main}]
  Consider setting $\kay = 13$ and $\eta \geq \kay$. Then,
  \[
   \eta\kay - 2\eta  - 4\kay \geq \eta(\kay - 6) \geq \frac{7}{13}\eta\kay
  \]
  We use $c_4$ to denote this constant of $\frac{7}{13}$, namely $c_4$
   satisfies
  \[
  c_{\textit{PC}}/\log^{\eta k - 2\eta - 4\kay}n \leq
  c_{\textit{PC}}/\log^{c_4 \eta\kay}n
  \]
  We can then pick a constant threshold $L$ and set $\kappa_i$ for all $i \leq L$ as
  follows:
  \[
  \kappa_1 = \log^{\kay^2}n, \kappa_2 =
  \log^{(2c_4)\kay^2}n, \cdots, \kappa_i = \log^{(2c_4)^{i - 1}\kay^2}n
  \]
  To solve $A_L$, we apply Lemma~\ref{lemma:uniformkappa}, which is
  analogous to setting $A_L, \dots, A_d$ uniformly.  The depth
  required in constructing these preconditioners is $O(m_d +
  \sum_{j=1}^L (2c_4)^{j -1}\kay^2)$, plus $O(m_d)$ for computing the
  inverse at the last level---for a total of $O(m_d) = O(m^{1/3})$.

  As for work, the total work is bounded by
\begin{align*}
&\!\!\!\!\!\!\sum_{i \leq d} m_i \prod_{1\leq j \leq i} \sqrt{\kappa_j}+
  \prod_{1\leq j \leq d}\sqrt{\kappa_j} m_d^2\\
  &=  \sum_{i < L} m_i \prod_{1 \leq j \leq i} \sqrt{\kappa_j}\\
  &\mbox{}\qquad+ \left( \prod_{1 \leq j < L} \sqrt{\kappa_j}\right)\cdot
 \left( \sqrt{\kappa_j} \sum_{i \geq L} m_i \prod_{L \leq j \leq  i} \sqrt{\kappa_j}+
    m_d^2\prod_{L \leq j \leq d} \sqrt{\kappa_j} \right)\\
  &\leq  \sum_{i < L} m_i \prod_{1 \leq j \leq i} \sqrt{\kappa_j}
  + \left(\prod_{1 \leq j < L} \sqrt{\kappa_j}\right) m_L \sqrt{\kappa_L}\\
  &= \sum_{i \leq L} m_i \prod_{1 \leq j \leq i} \sqrt{\kappa_j} \\
  &\leq  \sum_{i \leq L} \frac{m}{\prod_{j < i} \kappa_i^{c_4}}\prod_{1 \leq j \leq i} \sqrt{\kappa_j}\\
  &=   m \sum_{i \leq L} \frac{ \sqrt{\kappa_1} \prod_{2 \leq j \leq i}
    \sqrt{\kappa_{j - 1}^{2c_4}}  }{ \prod_{j < i} \kappa_i^{c_4}} \\
  &=  m L \sqrt{\kappa_1}
\end{align*}

The first inequality follows from the fact that the exponent of
$\log^{n}$ in $\kappa_L$ can be arbitrarily large, and then applying
Lemma~\ref{lemma:uniformkappa} to the solves after level $L$.  The
fact that $m_{i+1} \leq m_i \cdot O(1/\kappa_{i}^{c_4})$ follows from
Lemma~\ref{lem:precondition}.

  Since $L$ is a constant, $\prod_{1 \leq j \leq L} \in O(\polylog{n})$, so the total
  depth is still bounded by $O(m^{1/3+\theta})$ by
  Lemma~\ref{lemma:uniformkappa}.
\end{proof}


\section{Conclusion}
\label{sec:concl}

We presented a near linear-work parallel algorithm for constructing
graph decompositions with strong-diameter guarantees and parallel
algorithms for constructing $2^{O(\sqrt{\log n \log\log n})}$-stretch
spanning trees and $O(\log^{O(1)}n)$-stretch ultra-sparse subgraphs.
The ultra-sparse subgraphs were shown to be useful in the design of a
near linear-work parallel SDD solver.  By plugging our result into
previous frameworks, we obtained improved parallel algorithms for
several problems on graphs.

We leave open the design of a (near) linear-work parallel algorithm
for the construction of a low-stretch tree with polylogarithmic
stretch.  We also feel that the design of (near) work-efficient
$O(\log^{O(1)} n)$-depth SDD solver is a very interesting problem that
will probably require the development of new techniques.


\medskip
\section*{Acknowledgments}
  This work is partially supported by the National Science Foundation
  under grant numbers CCF-1018463, CCF-1018188, and CCF-1016799, by an
  Alfred P. Sloan Fellowship, and by generous gifts from IBM, Intel,
  and Microsoft.

\bibliographystyle{alpha}
\bibliography{../bib/ref,../bib/embedding,../bib/abbrev}


\end{document}